\documentclass[10pt,a4paper]{article}

\usepackage[ngerman,english]{babel}
\usepackage{a4wide}
\usepackage{graphics}
\usepackage{dsfont}
\usepackage{amsmath}
\usepackage{amssymb}

\usepackage{enumerate, amssymb, amsmath, amsthm}
\newcommand{\N}{\ensuremath{{\mathbb N}}}

\usepackage{tikz} 
\usepackage[T1]{fontenc}
\usepackage{lmodern}
\usepackage{units}

\newtheorem{satz}{Theorem}[section]

\newtheorem{prop}[satz]{Proposition}

\newtheorem{anmerk}[satz]{Remark}

\newcommand{\R}{\ensuremath{{\mathbb R}}}

\begin{document}
\begin{center}
{\Large{\textbf{\MakeUppercase{A method for pricing American options using semi-infinite linear programming}}}}
\end{center}
\begin{center}\vspace{.8cm}{\large{ S\"oren Christensen}} 

\large\textit{Christian-Albrechts-Universit\"at Kiel\let\thefootnote\relax\footnotetext{Mathematisches Seminar, Christian-Albrechts-Universität zu Kiel, Ludewig-Meyn-Str. 4
D-24098 Kiel, Germany, e-mail: christensen@math.uni-kiel.de }\\}
\end{center}

 We introduce a new approach for the numerical pricing of American options.
The main idea is to choose a finite number of suitable excessive functions (randomly) and to
find the smallest majorant of the gain function in the span of these functions. The resulting
problem is a linear semi-infinite programming problem, that can be solved using standard
algorithms. This leads to good upper bounds for the original problem. For our algorithms
no discretization of space and time and no simulation is necessary. Furthermore it is
applicable even for high-dimensional problems. The algorithm provides an approximation
of the value not only for one starting point, but for the complete value function on the
continuation set, so that the optimal exercise region and e.g. the Greeks can be calculated.
We apply the algorithm to (one- and) multidimensional diffusions and to Lévy processes, and
show it to be fast and accurate.\vspace{.8cm}

\vspace{.8cm}


\section{Introduction}\label{sec:intro}
Pricing American type options on multiple assets is a challenging task in mathematical finance and is important both for theory and applications. The problem to be solved is an optimal stopping problem. These problems play an important role in many other fields of applied probability, too. They also appear, for example, in mathematical statistics and portfolio optimization. Although a general theory is well developed (cf. e.g. \cite{ps}), the value and the optimal strategy in optimal stopping problems cannot be found explicitly in most situations of interest. Many approaches have been proposed for a numerical solution of optimal stopping problems in the last years.\\
For pricing the standard American options in the Black-Scholes market with one underlying the most prominent methods are algorithms based on backward induction, partial differential equation methods and integral equation methods, cf. e.g. \cite[Chapter 8]{De}. But these techniques are limited to low dimensional problems. Most techniques used today for more complex options are based on Monte Carlo simulation techniques that were developed in the last years, see \cite[Chapter 8]{Gla} for an overview. We only want to mention stochastic mesh- and regression based-methods, that are often combined with using a duality method.\\
A not that popular class of algorithms uses the linear programming approach. The basic idea is to approximate the underlying process by a Markov chain with a finite state space and to rewrite the resulting optimal stopping problem into a linear program and to solve this problem using standard techniques, cf. e.g. \cite{CS} and the references therein for infinite time horizon problems. By the curse of dimensionality this approach is limited to low dimensional problems.\\
Our approach is different in nature to all approaches described above. The result of this algorithm is an analytic approximation to the value function in the continuation set without using discretization. This basic idea is described in the following section and it is shown how semi-infinite linear problems come into play. In Section \ref{sec:cutting} we discuss how the well-known cutting plane algorithm can be used to solve such problems. In Section \ref{sec:diffusions} we motivate the further steps by discussing optimal stopping problems with infinite time horizon for one-dimensional diffusion processes. In Section \ref{sec:appendix} we give a theoretical explanation for the accuracy of the algorithm locally around the optimization point. The ideas described so fare are then applied to the more interesting case of American options on one and more assets with finite time horizon in Section \ref{sec:finite_time}. In this section the calculation of the Greeks and implied-volatility problems are also discussed. In Section \ref{sec:infinite_time} we show how the algorithm can be applied for multidimensional diffusions and Lévy processes with infinite time horizon. Summarizing the results we can say that our algorithm provides good upper bounds for the value. In Section \ref{sec:lower_bounds} we shortly discuss how it can be used to also obtain lower bounds. Finally we give a short conclusion in Section \ref{sec:conclusion}. 

\section{The approach}
We consider a Markovian problem of optimal stopping as follows:\\
Let $(X_t)_{t\geq 0}$ be a continuous time strong Markov process with state-space $E$, $g:E\rightarrow[0,\infty)$ measurable, $T\in(0,\infty]$ and $r\geq0$. We would like to maximize the expected value 
\[E_x(e^{-r\tau}g(X_\tau))\]
over all stopping times $\tau\leq T$ with respect to the underlying filtration for all starting points $x\in E$. If $T=\infty$ we say that we have an infinite time horizon, if $T<\infty$ we speak of a finite time horizon. For convenience  we assume $T=\infty$ in this section, then the value function does not depend on $t$; this is no real restriction, see \cite[Chapter I]{ps}; in this reference all the following basic facts can also be found.\\
We define the value function $v:E\rightarrow\R$ by
\[v(x)=\sup_\tau E_x(e^{-r\tau}g(X_\tau)).\]
If we know the value function $v$, then the optimal stopping problem is solved, but in most situations of interest it is not possible to give an explicit expression for $v$. From the theory for optimal stopping Markovian problems it is well-known that under minimal conditions the function $v$ can be characterized as the smallest $r$-excessive function w.r.t. $X$ that majorizes $g$, i.e. for a fixed $x_0\in E$ it holds that
\[v(x_0)=\inf\{h(x_0):h\mbox{~is $r$-excessive,}~h\geq g\}.\]
Here $r$-excessive functions are the class of functions, that correspond to the standard supermartingales for Markov processes, see e.g. \cite{ps} and the references therein. This formulation corresponds to the characterization of the value process as the smallest supermartingale dominating the gain process in the general setting. Unfortunately the space of $r$-excessive functions for a process $X$ is very wide in general, so that this characterization can be used  for an explicit determination of the value only in some very special settings.\\
Nonetheless using standard terms of optimization theory $v(x_0)$ can be seen as the solution to the following problem:
\begin{align*}
\min ! ~~~~~&h(x_0)\\
\mbox{s.t.~~~~}&h(x)\geq g(x)\mbox{~~~for all $x\in E$,}\\
~&h\mbox{~is $r$-excessive.}
\end{align*}
To see the direct connection to linear programming, let us rewrite the problem as follows: By Martin boundary theory -- cf. \cite{KW} --  under weak conditions on the process $(X_t)_{t\geq 0}$ each $r$-excessive function $h$ can be represented as 
\begin{equation}\label{int_repr}
h(\cdot)=\int_B k_b(\cdot)\pi(db),
\end{equation}
where $B$ is a compact space (called minimal Martin boundary), $k_b,b\in B,$ are the minimal $r$-excessive functions and $\pi$ is a measure on $B$. Therefore we see that the optimization problem described above can be seen as an linear infinite optimization problem:
\begin{align*}\label{LIP}
\min ! ~~~~~&\int_B k_b(x_0)\pi(db)\tag{LIP}\\
\mbox{s.t.~~~~}&\int_B k_b(x)\pi(db)\geq g(x)\mbox{~~~for all $x\in E$,}\\
~&\pi\mbox{~is a measure on $A$.}
\end{align*}
Here infinite means that we have infinitely many restrictions (since we always assume $E$ to be infinite) and we optimize over an infinite-dimensional space of measures. One standard way to treat this problem is to discretize the state space $E$. This leads to an ordinary linear programming problem, but this problem is solvable only for low-dimensional spaces $E$ by the curse of dimensionality.\\ 
The basic idea of our approach is to approximate the value of this linear infinite programming problem by reducing the problem to a semi-infinite linear programming problem by choosing a finite dimensional subspace of the measure space:

\begin{enumerate}
\item Fix $n\in\N$ and choose (in a suitable way) a finite subset $H:=\{h_1,...,h_n\}$ of $r$-excessive functions (equivalently choose $n$ measures $\pi_1,...,\pi_n$).
\item Solve the linear semi-infinite programming problem
\begin{align}\label{lp}
\min ! ~~~~~&\sum_{i=1}^n\lambda_i h_i(x_0)\tag{LSIP}\\
\mbox{s.t~~~~~}&\sum_{i=1}^n\lambda_i h_i(x)\geq g(x)\mbox{~~~for all $x\in E$}\notag
\end{align}
\end{enumerate}

\begin{anmerk}\label{rem:choos}
One way for choosing a suitable set $H$ is the following, that will be used in the examples below:\\
Take a subset $H'$ of the set of all $r$-excessive functions, that can be parametrized as $H'=\{h_a:a\in A\}$. Then choose random parameters $a^{(1)},...,a^{(n)}\in A$ (e.g. randomly with respect to a suitable probability distribution on $A$) and write $h_i:=h_{a^{(i)}},~ i=1,...,n$.
\end{anmerk}

One immediately obtains the following fact.

\begin{prop}
If $(\lambda_1^*,...,\lambda_n^*)$ is a solution of the linear semi-infinite programming problem (\ref{lp}), then $h^*:=\sum_{i=1}^n\lambda_i^* h_i(x)$ is an upper bound for $v(x)$ for all $x\in E$.
\end{prop}

\begin{proof}
The function $\sum_{i=1}^n\lambda_i^* h_i$ is an $r$-excessive function majorizing $g$. By the characterization of the value function as smallest $r$-excessive function majorizing $g$ we obtain the result.
\end{proof}

Note that although we only considered one special point $x_0$ for the optimization the function $h^*$ is is an upper bound for the value function $v$ on the whole domain $E$. As we will see later in many situations this is even a good upper bound on a huge neighborhod of $x_0$. But before we can apply this algorithm, the first question that arises is how linear semi-infinite programming problems of the type (\ref{lp}) can be solved:

\section{Cutting plane method for solving linear semi-infinite programming problems}\label{sec:cutting}

The theory for solving linear semi-infinite programming problems of the type (\ref{lp}) is well developed. A good overview is given in \cite{HK} where theory and algorithms are discussed. One of the key solution techniques is the so-called ''cutting plane algorithm''. It is based on solving a sequence of ordinary linear programming problems, where in each step one further constraint is added based on the results obtained so far. To be more precise in step $k$ one considers $\{x_1,...,x_k\}$ and solves

\begin{align}\label{LPk}
\min ! ~~~~~&\sum_{i=1}^n\lambda_i h_i(x_0)\tag{$LP_k$}\\
\mbox{s.t.~~~~~}&\sum_{i=1}^n\lambda_i h_i(x_j)\geq g(x_j)\mbox{~~~for each $j\in\{1,...,k\}$}\notag
\end{align}
If $(\lambda_1^{(k)},....,\lambda_n^{(k)})$ is an optimal solution to (\ref{LPk}) one chooses $x_{k+1}$ as a minimizer of the function 
\[x\mapsto \sum_{i=1}^n\lambda_i^{(k)} h_i(x)-g(x)\]
and uses the set $\{x_1,...,x_{k+1}\}$ in the next step. This sequence converges to the optimal solution under mild restrictions. This algorithm and a relaxed version are e.g. discussed in \cite{WFL}\\

We see that a semi-infinite linear programming problem is reduced to a sequence of standard linear programs. The background for the algorithm to work is the reduction theorem for semi-infinite linear programs, that states that the infinite restriction set $E$ can be reduced to a set of not more than $n$ points, cf. \cite[Chapter 10.2]{Ko}. The points $x_1, x_2,...$ are approximations to these points. Therefore the number of iteration steps in the cutting plane algorithm depends on the number $n$ of chosen excessive functions and not on the dimension of the underlying space. This is the key point for the applicability of our algorithm in higher dimensions. The standard approach to solve optimal stopping problems using linear programming is based on the discretization of the state space $E$, cf. e.g. \cite{CS}. Due to the curse of dimensionality this approach is limited to low dimensions.\\

For our approach to work we have to find suitable choices for $A$ and $H'$ (see Remark \ref{rem:choos}). As a motivation for the following sections we first consider one-dimensional diffusion processes with infinite time horizon:

\section{One-dimensional diffusions with infinite time horizon}\label{sec:diffusions}
One-dimensional diffusions have a wide range of applications e.g. in mathematical finance, mathematical biology, stochastic control and economics. We follow the definition given in \cite[Chapter VII.3]{RY} that is based on the work of Feller and Itô and McKean (cf. \cite{IM}), i.e. we assume that the process is a strong Markov process with continuous sample paths on an interval $E$. To prevent that the interval $E$ can be decomposed into disjoint subintervals from which $(X_t)_{t\geq0}$ cannot exit, we always assume that all diffusions are regular, that is
\[P_x(X_t=y\mbox{~for some~}t\geq 0)>0\mbox{~~~~~for all~}x\in int(I),~y\in I.\]
To find a suitable set $H$ of $r$-excessive functions we consider the functions
\[\psi_{1}(x)=\begin{cases}
  E_x(e^{-r\tau_a}\mathds{1}_{\{\tau_a<\infty\}}),  & x\leq a \\
  [E_a(e^{-r\tau_x}\mathds{1}_{\{\tau_x<\infty\}})]^{-1}, & x>a
\end{cases}\]
and 
\[\psi_{2}(x)=\begin{cases}
  [E_a(e^{-r\tau_x}\mathds{1}_{\{\tau_x<\infty\}})]^{-1},  & x\leq a \\
  E_x(e^{-r\tau_a}\mathds{1}_{\{\tau_a<\infty\}}), & x>a,
\end{cases}\]
for a fixed point $a\in int(E)$. These functions are called the minimal $r$-harmonic functions. Obviously $\psi_1$ is increasing and $\psi_2$ is decreasing. Furthermore they are positive, continuous and can be used to characterize the boundary behavior of $(X_t)_{t\geq0}$. For results in this direction we refer to \cite[Section 4.6]{IM}. All other positive $r$-harmonic functions are linear combinations of $\psi_1$ and $\psi_2$. 

In this section we want to study optimal stopping problems for one-dimensional diffusions with infinite time horizon. These problems can be solved analytically using different techniques; we only refer to \cite{m}, \cite{sa}, \cite{BL00}, \cite{dk} and \cite{CI2}. Nonetheless in many situations it can be helpful to use numerical methods. The following theorem guarantees that $H:=\{\psi_1,\psi_2\}$ is a reasonable choice for our algorithm to work well.

\begin{satz}\label{satz:value}
\begin{enumerate}[(a)]
Fix $x_0\in E$.
\item $v(x_0)$ is the value of the problem
\begin{align*}
\min ! ~~~~~&\lambda_1 \psi_1(x_0)+\lambda_2 \psi_2(x_0)\\
\mbox{s.t~~~~~}&\lambda_1 \psi_1(x)+\lambda_2 \psi_2(x)\geq g(x)\mbox{~~~for all $x\in E$}
\end{align*}
\item If $(\lambda_1,\lambda_2)$ is a solution to the problem above, then $v(x)=\lambda_1 \psi_1(x)+\lambda_2 \psi_2(x)$ for all $x$ in the connection component of $x_0$ in the continuation set.
\item If an optimal stopping time exists, then 
\[\tau_x=\inf\{t\geq0: g(X_t)=\sum_{i=1}^2\lambda_i \psi_i(X_t)\}\]
is optimal under $P_x$.
\end{enumerate}
\end{satz}
\begin{proof}
\begin{enumerate}[(a)]
\item Two proofs based on different methods can be found in \cite[Theorem 4.2]{HS} and in \cite[Corollary 2.2]{CI2}.
\item This is a general fact, see Proposition \ref{prop:maximumprinc} in the following section.
\item If an optimal stopping time exists, then by the general theory the smallest is given by $\tau^*=\inf\{t\geq0:v(X_t)=g(X_t)\}$. Since $(X_t)_{t\geq 0}$ has continuous sample paths the assertion holds by the previous point.
\end{enumerate}
\end{proof}

We consider $H:=\{\psi_1,\psi_2\}$. The theorem states that the value of the linear semi-infinite programming problem (\ref{lp}) is equal to $v(x_0)$ (and not only an upper bound). 
Our algorithm provides an accurate way for solving these problems using our approach. It is very easy to implement in every common language. This and all following examples were implemented in Matlab on a standard PC with 1.3 GHz. We used the cutting plane method to solve (\ref{lp}). This works fine and gives the results after some steps of iteration.\\
As an example we consider the gain function $g(x)=x^2$ for a standard Brownian motion $X$. Using our algorithm after 5 steps of iteration the linear semi-infinite programming problem reaches the solution $v(0)=5.322$ and one furthermore obtains $v(x)=2.661\psi_1(x)+2.661\psi_2(x)$ for $x\in[-4.618,4.618]$, where $\psi_1(x)=e^{0.447x},\psi_2(x)=e^{-0.447x}$. Moreover the optimal stopping time is $\inf\{t\geq 0: X_t\not\in [-4.618,4.618]\}$. A graphical illustration can be found in the following figure.
\begin{figure}[ht]
\begin{center}
\includegraphics[width=10cm]{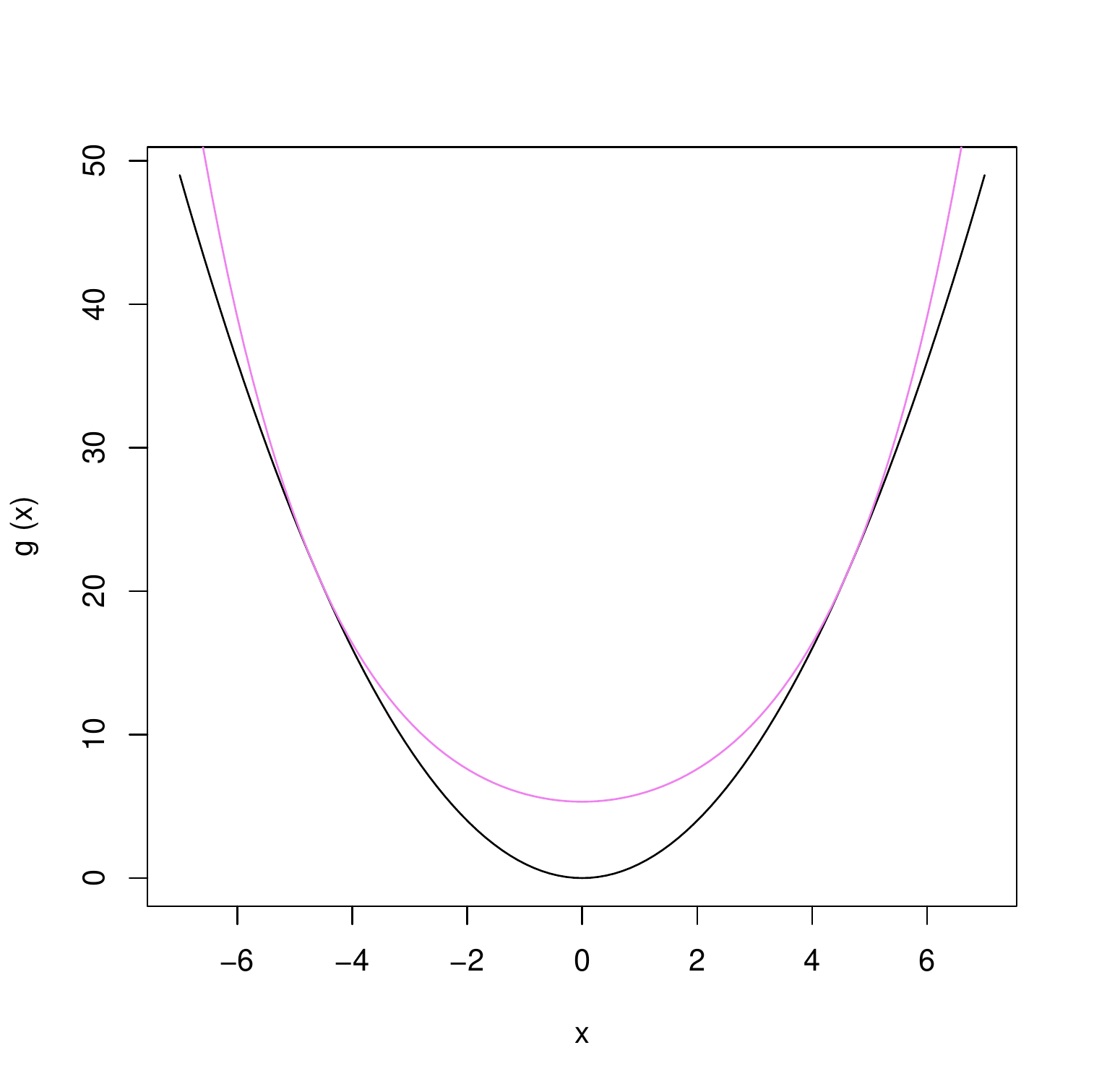}
\caption{Example for one-dimensional diffusions and infinite time horizon}
\end{center}
\end{figure}

This example is of course not that impressing since optimal stopping problems of this type can -- in many cases -- even be solved analytically by standard techniques such as a free boundary approach. But nonetheless this example is instructive for dealing with other problems. We can summarize the results as follows:
\begin{itemize}
\item In the definition of the set $H$ one can restrict oneself to $r$-harmonic functions (instead of general $r$-excessive ones). 
\item Optimizing for one point $x_0$ in the continuation set yields the value function for the whole connection component of the continuation set containing $x_0$.
\item One also obtains the optimal stopping time.
\end{itemize}

\section{Approximation in the connection component of the continuation set}\label{sec:appendix}
Next we give the theoretical background for the observation that using an $r$-harmonic function $h$ as an approximation to the value function in a fixed point $x_0$ yields a good approximation of the value function on the connection component of the continuation set containing $x_0$ for a wide class of Markov processes (see Theorem \ref{satz:value}(b) above and the numerical results in the next sections):\\
Let $C$ denote the connection component of the continuation set that contains $x_0$. We assume that the underlying Markov process fulfills a strong maximum principle on $C$, i.e. we assume that each function that is $r$-harmonic and attains its non-negative maximum in $C$ is constant. This principle is well known for certain processes as diffusions under mild conditions, cf. e.g. \cite[p. 84]{Pi}, and for further classes of processes. Under these assumption we have:

\begin{prop}\label{prop:maximumprinc}
Let $h$ be an $r$-harmonic majorant of $v$ with $h(x_0)=v(x_0)$.\\
Then $h\rvert_C=v\rvert_C$.
\end{prop}

\begin{proof}
Since $v$ is $r$-harmonic in $C$ so is $\tilde{h}=h-v$. Furthermore $\tilde{h}\geq0$ and $\tilde{h}(x_0)=0$. Therefore by the maximum principle $\tilde{h}$ is constant in $C$, i.e. $h\rvert_C=v\rvert_C$.
\end{proof}
Now we apply the ideas obtained so fare to the more interesting situation of finite time horizon and multidimensional diffusions:

\section{Time-dependent gain}\label{sec:finite_time}
In many applications the gain also depends on time, e.g. in mathematical finance one often considers problems with a finite time horizon. In this case the value function is an space-time $r$-excessive functions. Before we come to applications we first discuss how the transition densities come into play. To this end we use the integral representation as given in (\ref{int_repr}) for this case:\\
For a standard Brownian motion it is well known that each excessive function can be written as an integral taken over the densities of the Gaussian semigroup (cf. \cite{Si}). More recently this result was extended to a much more general setting in \cite{Ja}. Different results are given there. We only state the following special fact, that is useful for us:\\

\textit{Under some mild technical assumptions if the underlying transition semigroup is a convolution semigroup on a locally compact Abelian group $E$, then each each space-time excessive function $h$ has the representation
\[h(t,x)=\int 1_{(s,\infty)}(t)p_{t-s}(x,y)\pi(dy,ds),~~x\in E,t>0,\]
where $\pi$ is a measure on $E\times[0,\infty)$ and $p_t(x,y)$ is a suitably chosen density of the semigroup.}\\
Standard examples for densities, that are often used are the standard $d$-dimensional Brownian motion, where
\[p_t(x,y)=\frac{1}{\sqrt{2\pi t}^d}\exp \left(-\frac{|x-y|^d}{2t}\right)\]
and for Cauchy processes, where
\[p_t(x,y)=\frac{at}{(t^2+|x-y|^2)\frac{n+1}{2}},~~~~a=\nicefrac{\Gamma\left(\frac{n+1}{2}\right)}{\left(\pi\frac{n+1}{2}\right)}.\]
With this theoretic result in mind we can treat the well known examples from mathematical finance:

\subsection{American put in the Black-Scholes model}\label{put_endlich}

\begin{table}
\begin{center}
\begin{tabular}[h]{|l||c|c|c|c|}
  \hline
  $x$ & $v(x)$& RLP & |$v(x_0)$-RLP|&time\\
  \hline\hline
  $80$ & 21.606 & 21.615 &0.009&\\
	\hline
	$90$ & 14.919 & 14.923 &0.004&\\
	\hline
	$100$ & 9.946 &  9.951 &0.005& 9.4s\\
	\hline
  $110$ & 6.435 & 6.439 &0.004&\\
	\hline
	$120$ & 4.061 & 4.064 &0.003&\\
		\hline
\end{tabular}
\end{center}
\caption{Approximation of the ``true'' value $v(x)$ (taken form \cite{AC}) for the American put problem on one asset with time horizon $T=0.5$. RLP denotes the values for our algorithms. We applied the algorithm for the optimization point $x_0=100$ and obtained the other values from this optimization as described above. The parameters are $d=1, K=100, r=0.06, T=0.5, \sigma=0.4$.}
\end{table}

\begin{figure}[ht]\label{fig:value3d}
\begin{center}
\includegraphics[width=13cm]{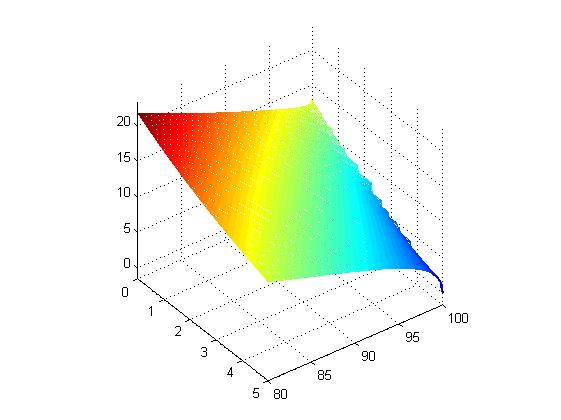}
\caption{Graph of the approximated value function $h^*$}
\end{center}
\end{figure}

As an example we consider a market Black-Scholes-market where the asset price process $X$ is a geometric Brownian motion under the risk neutral measure, that solves 
\[dX_t=rX_tdt+\sigma X_tdW_t\]
for a Brownian motion $W$. Although our approach is also applicable for other gain functions, in this subsection we concentrate on the fair price for an American put on $X$ with strike $K$ and maturity $T$ as given by
\[v(t,x)=\sup_{\tau\leq T-t}E_{(t,x)}(e^{-r\tau}(K-X_{t+\tau})^+),\]
since this example is well studied from different points of view. No closed form solutions are known for this problem, but many numerical methods are developed, cf. e.g. \cite{De} for an overview.\\
The transition density $p$ is given by
\[p_t(x,y)=\frac{\sigma}{2\sqrt{2\pi t}}(xy)^{-\nu}\exp\left(-\frac{\sigma^2\nu^2t}{2}-\frac{(\log(y)-\log(x))^2}{2\sigma^2t}\right),\]
where $\nu=\mu/\sigma^2-1/2$, cf. \cite{BS}, p. 132. The first idea to apply our algorithms is now to take these densities. But one sees that $p$ has a singularity for $t=0$. Therefore these densities are no good choice, since linear combinations cannot dominate the gain function. Therefore we take integrated versions of the density. The easiest such functions are  the $r$-harmonic function given by
\begin{align*}
h_a(t,x)&=E_{(t,x)}(e^{-r(T-t)}\mathds{1}_{\{X_T\geq a\}})\\
&=e^{-r(T-t)}\Phi\left(-\frac{\log(x/a)+(r-\sigma^2/2)(T-t)}{\sigma\sqrt{T-t}}\right)
\end{align*}
for $a \in(0,\infty)=:A$. Using these functions we can apply our algorithm and compare the results to prices taken from \cite{AC}. The results can be found in the following table. We obtained the data in the following way:\\
We use our approach with the starting value $x_0=100, t_0=0$ and $n=100$ and choose the parameters $a_1,....,a_n$ according to a uniform distribution on $[0,100]$. Applying the optimization takes around 10 second. Then we obtained the value at $x_0=100$ and $t_0=0$ and obtained parameters $\lambda_1^*,...,\lambda_n^*$ such that the function $h^*:=\sum_{i=1}^n\lambda_i^* h_{a_i}$ is an upper bound of the value function $v$. Then we used $h^*$ to get the upper bounds for other starting values by just evaluating this function at the desired point $(t,x)$. Although $h^*$ is optimized for the point $(t_0,x_0)=(0,100)$ the upper bounds for the other points are very good too for other starting prices as shown in the table. Hence we only need one approximation for all time horizons and starting values in the connected component of the continuation set, see Figure \ref{fig:value3d}. Hence we have found an analytically given function that is a good approximation to the value function on a huge subset of the time-space domain. With the results of Section \ref{sec:appendix} in mind this is not surprising.

\subsection{American min put on $d$ assets}
As discussed in the introduction it is much more challenging to consider multiple underlyings, i.e. the case that $X=(X^{(1)},...,X^{(d)})$ is a diffusion in a subset of $\R^d$. As an example we consider the multi-dimensional Black-Scholes market, i.e. $X^{(1)},...,X^{(d)}$ are geometric Brownian motions with fixed correlations of the underlying Brownian motions. One benchmark example in the literature is a put option on the minimum of the assets in a Black-Scholes market, i.e.
\[g(x_1,...,x_d)=\left(K-\min_{i\in\{1,...,d\}}x_i\right)^+.\]
We compare our results to the numerical results given in \cite[Section 4.2]{Ro}. With the same motivation as for one underlying we could choose the set $H'$ of $r$-harmonic functions to consist of the functions
\[h_a(t,x)=E_{(t,x)}\left(e^{-r(T-t)}\mathds{1}_{\{X_T^{(1)}\geq a_1,...,X_T^{(d)}\geq a_d\}}\right),\]
where $a=(a_1,...,a_d)\in(0,\infty)^d$. For highly correlated component processes and high dimensions the evaluation of these expectations takes much computational time. In these cases it is more reasonable to take the prices of European exchange options between the assets, since these integrals can be computed explicitly. Let us remark that all reasonable choices we tried lead to good results. We again used the point $x_0=(100,...,100)$ as the optimization point. Exact results and computational time can be found in the following tables. Summarizing we can say
\begin{itemize}
\item Optimizing for one special starting point gives very accurate approximations of the value function in the continuation set.
\item The same is true for varying time horizons.
\item The algorithm also works for large dimensions (e.g. $\geq 10$), where normally only Monte-Carlo methods are applicable.
\end{itemize}

\begin{table}\label{tab:minput2d}
\begin{center}
\begin{tabular}[ht]{|l||c|c|c|c|}
  \hline
  $x$ &   interval in \cite{Ro}&RLP &Comp. time\\
  \hline\hline
  $(80,80)$ &[38.01~,~38.35] & 38.30 &\\
	\hline
	$(80,100)$ &[32.23~,~32.60]& 32.28 &\\
	\hline
	$(80,120)$  &[28.54~,~29.01]&   28.58 &\\
	\hline
  $(100,80)$  & [33.34~,~33.59]& 33.53&\\
	\hline
	$(100,100)$   &[25.81~,~26.02]& 25.86& $42s$\\
		\hline
	$(100,120)$ &[20.75~,~21.05]&20.73&\\
		\hline
	$(120,80)$  &[31.21~,~31.31]& 31.30&\\
		\hline
	$(120,100)$  &[22.77~,~22.83]& 22.80&\\
		\hline
	$(120,120)$  &[16.98~,~16.98]& 16.99&\\
		\hline
\end{tabular}
\caption{This table states the results for the min-put problem in a Black-Scholes market with two assets and parameters $K=100, r=0.06, T=0.5, \sigma_1=0.4,\sigma_2=0.8, n=150$. $x$ denotes the starting value. We applied our algorithm with starting vale $(100,100)$ and obtained the further approximations reported in column ``RLP'' by evaluating the approximation $h^*$. For comparison we also state the values from \cite{Ro}.}
\end{center}
\end{table}
\begin{table}
\begin{center}
\begin{tabular}[ht]{|l||c|c|c|c|}
  \hline
  $d$ & interval in \cite{Ro}&  RLP &Comp. time\\
  \hline\hline
  $2$  &[24.87~,~25.16]& 24.93 &41 s\\
	\hline
	$3$ &[31.21~,~31.76]& 31.41 & 72 s\\
	\hline
	$4$  &[35,72~,~36.28]&   36.01& 115 s\\
	\hline
  $5$  & [39.01~,~39.47]& 39.21& 103 s\\
	\hline
	$10$   &[47.99~,~48.33]& 48.01& 324 s\\
		\hline
	$15$ &[52.23~,~52.14]& 52.10 & 612 s \\
		\hline
\end{tabular}
\caption{In this table we state the results in the same setting as in Table \ref{tab:minput2d} for dimensions 2 to 15. The parameters now are $\sigma_i=0.6$, $T=0.5$, $K=100$, $r=0.06$}
\end{center}
\end{table}

\subsection{Exercise boundary}
Using our approximation of the value function we can also approximate the optimal exercise boundary: For the true value function $v$ and each $t\in[0,T]$ the exercise boundary is characterized as the largest 0 of $v(t,\cdot)-g(\cdot)$. Therefore for an approximation we use the minimizer $b(t)$ of $(h^*-g)(t,\cdot)$ for each $t\in[0,T]$. A priori it is not clear that this approximation will be accurate even if $h^*$ is a good approximation of the value function, since $b(\cdot)$ is very sensitive to the shape of $h^*$ in both variables $x$ and $t$. Nonetheless this approximation is very good as shown in the figure below. There we compared the boundary to the approximations by standard methods given in \cite{LS}. We would like to underline that the approximations is even good for small time horizons.

\begin{figure}[ht]
\begin{center}
\includegraphics[width=10cm]{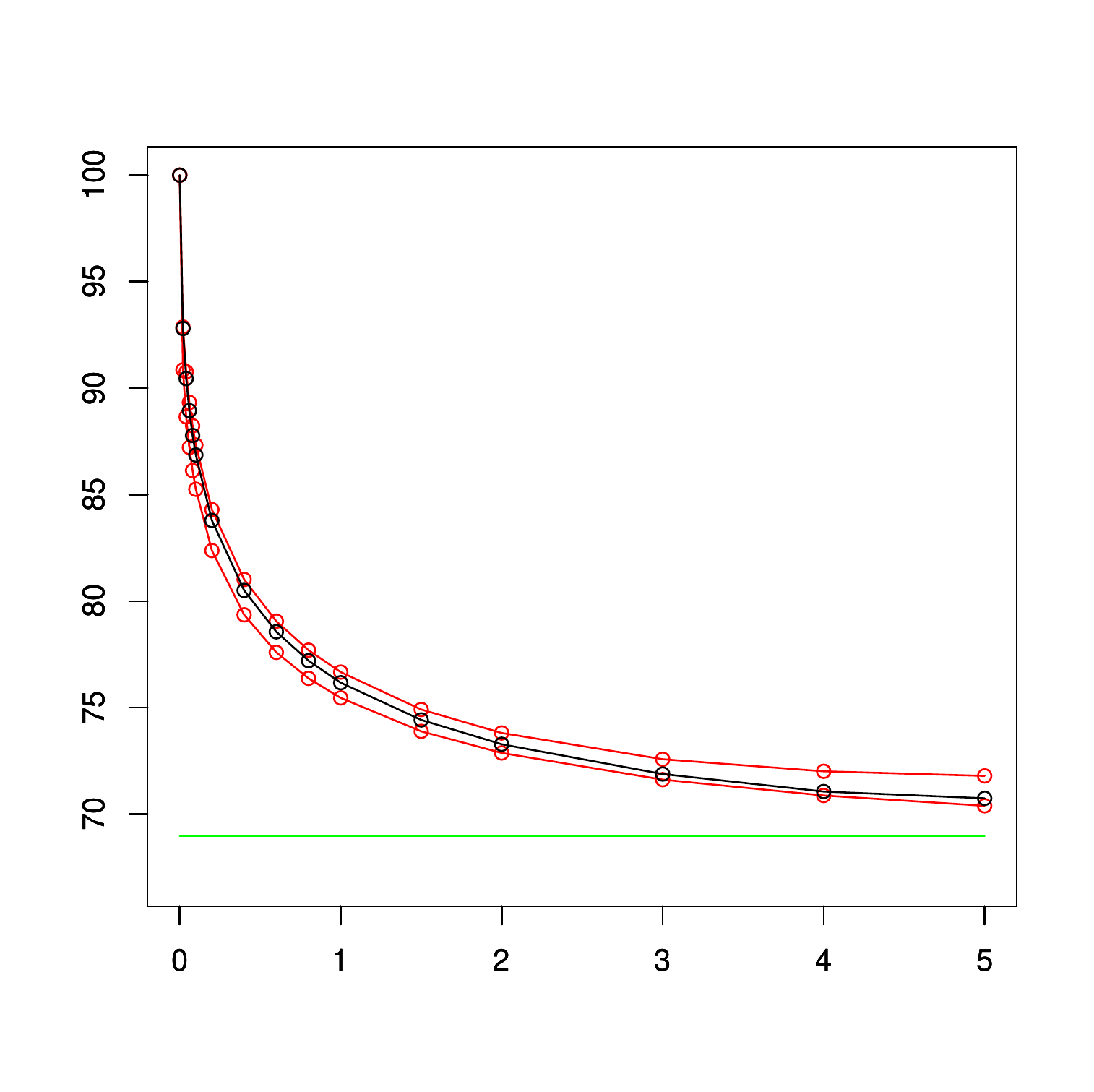}
\caption{The black graph is the approximated stopping boundary $b(\cdot)$ for parameters $T=1, K=100, r=0.1$ obtained by our algorithm. The red graphs are approximations taken from \cite{LS}, where the upper one is obtained by the PSOR-method and the lower one is the analytical approximation given by Zho. For comparison the green line is the optimal stopping boundary for the perpetual American put.}
\end{center}
\end{figure}

\subsection{Calculation of the Greeks}
For risk management and hedging the Greeks (sensitives) of the option play a major role. The Delta of the option -- i.e. the derivative with respect to the asset price -- is of special interest. Using Monte-Carlo techniques it is not straightforward to calculate it, see \cite[Chapter 7]{Gla}. Using our method we obtain an approximation $h^*(t,x)$ to the value function on the continuation set as a function of space and time. Therefore we can calculate the Delta and the Theta of the option by simply taking derivatives of $h^*$ with respect to $x$ resp. $t$. Comparing with the results in the literature yields that these estimates are very accurate. 

\subsection{Calculation of implied volatilities}
Another important topic in the valuation of American options is finding the implied volatilities for a given market price $v_0$. From a first view our approach does not seem to be reasonable for  this question, because the value function for one special volatility first of all does not give information about the values for other volatilities. In the following we discuss how this important topic can nonetheless be dealt with:\\
For a fixed starting volatility $\sigma_1$ we can approximate the price of the asset using our algorithm by finding an approximation $h_1(0,x;\sigma_1)$ and can compare this result with the market price $v_0$. Note that this expression gives an explicit function of $\sigma_1$, but for $\sigma\not=\sigma_1$ it is not clear if $h_1(0,x;\sigma)$ is an accurate approximation of the price in the model with volatility $\sigma$. But nonetheless we can solve the equation
\[h_1(0,x;\sigma)=v_0\]
for $\sigma$. Denote the solution by $\sigma_2$. Using our approach again we find a new approximation $h_2(x,0;\sigma_2)$ that can be used to determine $\sigma_3$ and so on. In a general setting there is no hope to prove convergence of this sequence to the implied volatility, but nonetheless in our examples one obtains a very accurate approximation after three or four steps of iteration even for starting volatilities that are fare away from the correct value. This leads to a very easy to implement method. Let us emphasize that there is no theoretical justification for the approach to work, but nonetheless it seems to work very well.

\section{Infinite time horizon}\label{sec:infinite_time}
After discussing finite time horizon problems in the previous section now we want to discuss the case of an infinite time horizon. For practical questions in financial markets this case is not so important; perpetual options are only used as a bound for finite time problems. Nonetheless for other applications such as sequential statistics and portfolio optimization numerical solutions are of importance. Most other numerical methods cannot be applied to these problems, since a discretization of an infinite time horizon would be necessary. Exception are the Forward Improvement Iteration algorithm discussed in \cite{I} and the results of \cite{CS}.

\subsection{Multidimensional diffusions}
In the following we are interested in the case that $X=(X_t)_{t\geq0}$ is a diffusion process with state space $E\subseteq \R^d, d\geq 1$. For applying our approach we first have to take a suitable subset of the class of $r$-excessive functions w.r.t. $X$ with a suitable parametrization. As in the case of finite time horizon we propose to choose a class of $r$-harmonic functions on $E$. In this setting we propose to take the minimal $r$-harmonic functions; that are the extreme points of the set of all $r$-harmonic functions on $E$ and can be characterized using the Martin boundary. There is a one-to-one correspondence between minimal $r$-harmonic functions.\\
Next we give an examples of interest in mathematical finance where explicit results can be obtained.

\begin{prop}\label{harm-brown}
Let $X$ be a $d$-dimensional Brownian motion on $\R^d$ with covariance matrix $(\sigma_{ij})$ and drift $\mu=(\mu_1,..,\mu_d)$, i.e. the generator of $X$ is given by 
\[L=\frac{1}{2}\sum_{i,j}\sigma_{ij}\frac{\partial^2}{\partial x_i\partial x_j}+\sum_{i}\mu_i\frac{\partial}{\partial x_i}.\]
Write
\[A=\left\{a\in\R^d: \frac{1}{2}\sum_{i,j=1}^d\sigma_{ij}a_ia_j+\sum_{i=1}^d\mu_ia_i-r=0\right\}.\]
Then 
\[\{x\mapsto \exp(a\bullet x):a\in A\}\]
is the set of all minimal positive $r$-harmonic functions, where $\bullet$ denotes the usual scalar product.
\end{prop}
This result is well-known, a discussion for this situation can be found in \cite{CI2}. As a standard example for a multidimensional problem we consider a perpetual American put option on on index, i.e. on a linear combination of assets. This means we consider $d$ (correlated) Brownian motions $X^{(1)},...,X^{(d)}$ with drifts $\mu_1,...,\mu_d$ and covariance structure $(\sigma_{ij})_{i,j=1}^d$. We interpret $e^{X^{(1)}},...,e^{X^{(d)}}$ as $d$ assets in a Black-Scholes market. Our gain function is given by 
\[g(x_1,...,x_d)=(K-\sum_{i=1}^d\alpha_ie^{x_i})^+\mbox{~~~for all~}(x_1,...,x_d)\in\R^d.\]
Here $\alpha_1,...,\alpha_d$ are positive weight parameters. These options were considered from different points of view, see e.g. \cite{P}. \\
To use our approach Proposition \ref{harm-brown} suggests to take 
\[A:=\left\{a\in\R^d: \frac{1}{2}\sum_{i,j=1}^d\sigma_{ij}a_ia_j+\sum_{i=1}^d\mu_ia_i-r=0\right\}.\]
and
\[H':=\{x\mapsto \exp(a\bullet x):a\in A\}.\]
The set $A$ is an ellipsoid and we can choose the random parameters $a^{(1)},...,a^{(n)}$ from a uniform distribution on $A$.\\
One cannot expect to obtain explicit results for this problem, so that we have to compare our results to other approximative results. For this reason we use the forward improvement iteration algorithm discussed in \cite{I}. This algorithm can be applied easily in dimension $d=2$, so that we use it for our comparison there. We used the forward improvement iteration algorithm with a discretization of $[0,2]\times [0,2]$ in $100\times 100$ points.
The results are given in the following table. Here the approximation of $v(x)$ using our approach is denoted by RLP, the results by the forward improvement iteration algorithm is denoted by FII.

\begin{table}
\begin{center}
\begin{tabular}[ht]{|l||c|c|c|c|c|}
  \hline
  $x_0$ & $g(x_0)$& FII & RLP & stopping point& |RLP-FII|\\
  \hline\hline
	$(0.7,0.2)$ &6.764&6.764& 6.764&yes&0\\
	\hline
  $(0.7,0.7)$ &5.972&5.976& 5.977&yes&0.005\\
	\hline
  $(1,1)$ &4.563&4.894& 4.944&no&0.05\\
	\hline
	$(1.4,0.6)$ &4.122&4.790& 4.778&no&0.012\\
	\hline
\end{tabular}
\caption{Results for infinite time horizon and a put on an index. We used the parameters $\mu_1=\mu_2=0, \sigma_{ij}=\delta_{ij}, r=0.1, K=10, \alpha_1=\alpha_2=1, n=30$. With FII we denote the value obtained by the forward improvement iteration algorithm and by RLP the values obtained by our approach. The optimization toke around 12 seconds.}
\end{center}
\end{table}

Although the forward iteration improvement algorithm is limited to low dimensions, our algorithm is not. It is no problem to to apply it to high dimensional problems.\\
%
%

\subsection{Lévy processes with infinite time horizon}\label{levy}
Lévy processes are an important class of jump processes that can be used in many fields of application such as insurance and finance. Optimal stopping problems with infinite time horizon involving Lévy processes were studied from different points of view in the last years. For these problems overshoot plays a fundamental role. This leads to certain problems for an explicit solution. For certain gain functions -- such as power functions and functions of put-/call-type -- semi-explicit solutions were obtained in the terms of the running maximum resp. minimum of the process, cf. \cite{Mo}, \cite{NS2}, \cite{ms2} and \cite{CI}. \\
To use our approach we again use the following potential-theoretic representation of $r$-excessive functions: \\
As usual we define the resolvent kernel $U^r$ by 
\[U^r(x,A)=E_x\left(\int_0^\infty e^{-rt}1_{\{X_t\in A\}}dt\right),~~A\mbox{~measurable},~x\in\R^d\]
and we assume that $U^r(x,\cdot)$ is absolutely continuous with respect to the Lebesgue measure for all $x\in\R^d$. See \cite[Chapter I.3]{b} for a characterization; all the next facts can also be found in this reference.\\
In the above situation there exists a unique measurable function $h:\R^d\rightarrow[0,\infty]$ such that $h(\cdot-x)$ is a Lebesgue density of $U^r(x,\cdot)$ for each $x\in\R^d$ and $y\mapsto h(-y)$ is $r$-excessive. For certain processes $h$ can be calculated explicitly.
Now we can formulate the important representation result in the spirit of equation (\ref{int_repr}):
\begin{prop}
Any integrable $r$-excessive function $w$ can be represented as 
\[w(x)=\int_{\R^d} h(a-x)\pi(da),~x\in\R^d,\]
where $\pi$ is a unique finite measure on $\R^d$. 
\end{prop}

\begin{figure}[h]
\begin{center}
\includegraphics[width=10cm]{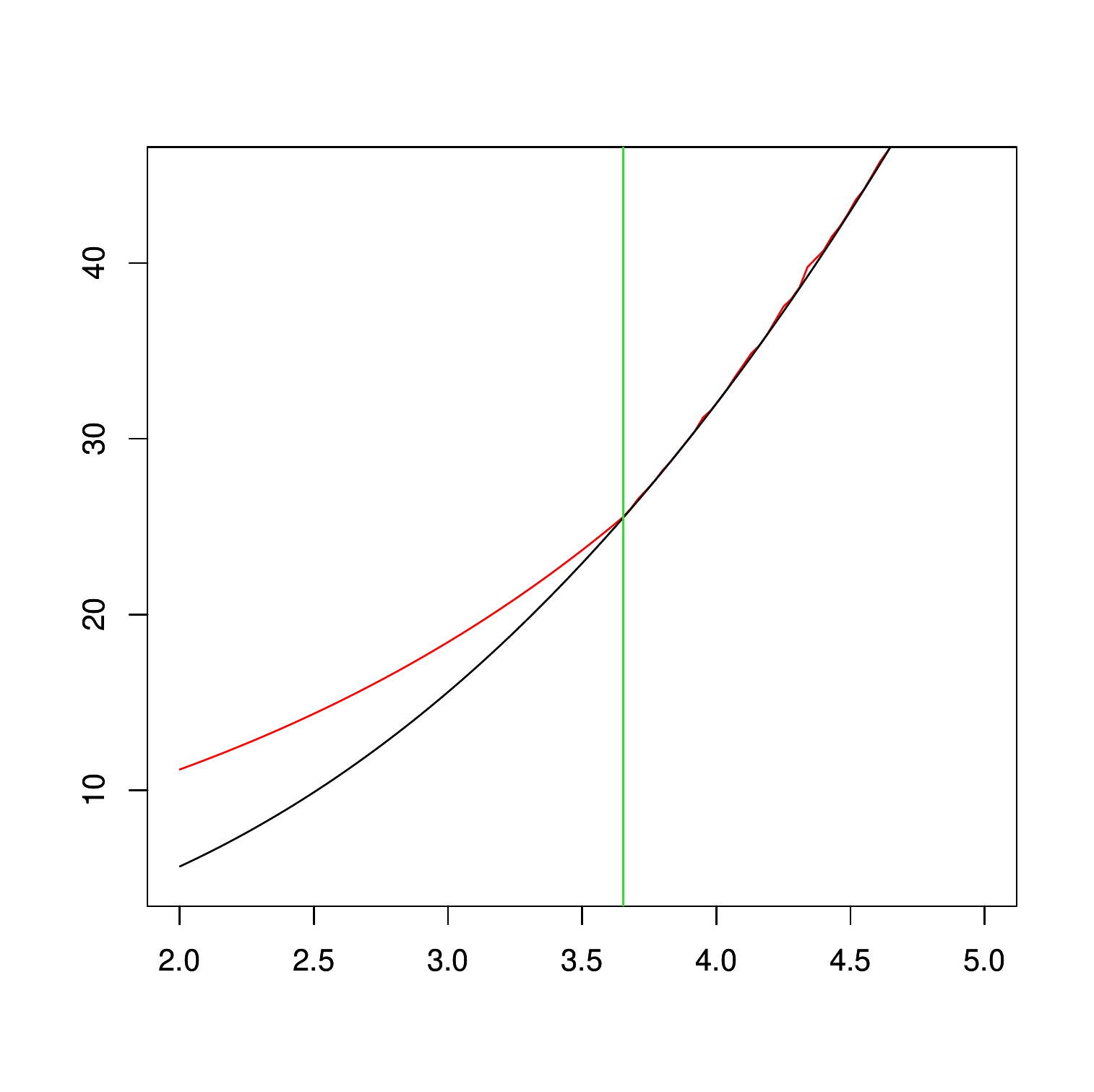}\label{fig:levy}
\caption{Gain function (black) and approximated value function (red) in the Lévy example for parameters $r=2, \alpha=1, \rho=0.5, \lambda/c=-0.5, \gamma=2.5$.}
\end{center}
\end{figure}

\begin{anmerk}
For $r=0$, i.e. for problems without discounting, an analogous result holds, if the Lévy process is assumed to be transient, see \cite{ms2} for the case $d=1$.
\end{anmerk}

Now we can use the algorithm described above by taking 
\[A=\R^d,~~H'=\{y\mapsto h(a-y):a\in A\}.\]
As an example we consider the Novikov-Shiryaev problem, i.e. we use $g(x)=(x^+)^\gamma$, $\gamma\geq1$ as the gain function. This problem was completely solved semi-explicitly in \cite{NS} and \cite{NS2}.\\
To check our numerical approach for Lévy processes we would like to compare our numerical results to explicit results. To this end we assume $X$ to be a compound Poisson process with drift and positive exponential jumps, i.e. $X$ has the form
\[X_t=ct+\sum_{i=1}^{N_t}Y_i,~~t\geq0,\]
where $c<0$, $(N_t)_{t\geq 0}$ is a Poisson process with intensity $\lambda$ and $(Y_i)_{i\in\N}$ is a sequence of independent $Exp(\alpha)$-distributed random variables. In this setting an explicit solution was obtained in \cite{ms2}. In this case the Green function $h$ is given by
\[h(x)=\begin{cases}
  A_2e^{\rho x},  & x\leq 0 \\
  -A_1, & x>0,
\end{cases}\]
where $\rho:=\alpha+\lambda/c>0$, $A_1:=\alpha/(\lambda+c\alpha)$ and $A_2:=\lambda/(a\lambda+a^2\alpha)$, see \cite[Section 5]{ms2}. \\
We applied our algorithm with $n=150$, and choosing the parameters $a^{(1)},...,a^{(n)}$ according to a uniform distribution on the interval $[0,20]$ turned out to be a reasonable choice. The computational results are given in Figure \ref{fig:levy}. Using these $r$-excessive functions we obtain a good approximation not only on the continuation set, but also on the optimal stopping set. The valuation of American options with finite time horizon in Lévy markets will be discussed in detail in a forthcoming paper.


\section{Lower bounds of the value function}\label{sec:lower_bounds}
As explained above our method immediately leads to good upper bounds of the value function. This is indeed the important contribution of our approach since most Monte-Carlo methods leads to good lower bounds, but nonetheless to deal with new problems one also would like to obtain lower bounds for the value. In this section we discuss how this can be realized using our approach. \\
For easy examples like the American put in the Black-Scholes market the method also gives an approximation to the stopping boundary. Using the stopping time associated with this boundary leads to very good lower bounds.\\
For more complex examples the idea is to use the approximation of the value function $h^*=\sum_{i=1}^n\lambda_i^*h_{i}$. The first idea is to choose the optimal stopping time 
\[\tau^*=\inf\{t\geq 0: g(X_t)\geq v(t,X_t)\}\]
and to substitute $v$ by $h^*$. But since $h^*\geq g$ and $h^*$ is just an approximation to $v$ this stopping time does not seem to be appropriate. Instead we choose $\epsilon>0$ and take 
\[\tau'=\inf\{t\geq 0: g(X_t)+\epsilon\geq v(t,X_t)\},\]
then it is well known that $\tau'$ is $\epsilon$-optimal in the sense that 
\[E_(t,x)(e^{-r\tau'}g(X_{\tau'}))\geq v(t,x)-\epsilon\]
under minimal conditions. Now we take 
\[\tau_0:=\inf\{t\geq 0: g(X_t)+\epsilon\geq h^*(t,X_t)\}\]
and use 
\[h_*(t,x):=E_{(t,x)}(e^{-r\tau_0}g(X_{\tau_0}))\]
is a lower bound of $v(t,x)$. In most problems we cannot find analytical expressions for the expectation on the right hand side, but it can be approximated using Monte Carlo techniques. In our examples the lower approximations were quite near the upper ones.\\
Another approach is to use the $r$-harmonic function $h^*$ for variants of other methods like the Longstaff-Schwartz algorithm. This will be discussed by the author in a forthcoming paper.

\section{Conclusion}\label{sec:conclusion}
As a conclusion let us summarize the important properties of the approach described above:
\begin{itemize}
\item The algorithm is based on reducing the ILP-problem connected to optimal stopping to a SILP-problem by choosing finitely many $r$-excessive functions.
\item No discretization of space and time and no Monte-Carlo-elements are necessary.
\item The algorithm is very easy to implement in every common language (1 page of programming code!).
\item Optimizing for one special starting point gives very accurate upper bound of the value function in the continuation set.
\item The same is true for varying time horizons.
\item The algorithm also works for large dimensions (e.g. $\geq 10$), where apart from it only Monte-Carlo methods are applicable.
\item The Greeks can be found immediately.
\item Implicit volatilities can be calculated.
\item The algorithm can be used for infinite time horizons, too.
\end{itemize}

\bibliographystyle{plain}
\bibliography{numerik}

\end{document}